\documentclass[10pt]{article}

\usepackage{amssymb}
\usepackage{amsfonts}
\usepackage{amsmath}
\usepackage{amscd}
\usepackage{amsthm}
\usepackage{setspace}
\usepackage{enumerate}
\usepackage{geometry}

\theoremstyle{plain}
\newtheorem{thm}{Theorem}

\newtheorem{lem}[subsection]{Lemma}

\theoremstyle{definition}
\newtheorem{rem}{Remark}
\newtheorem{example}{Example}
\newtheorem{defi}{Definition}
\newtheorem{remark}{Remark}

\newcommand{\CC}{{\mathbb C}}

\usepackage{hyperref}

\title{Spectral curve duality beyond the two-matrix model}
\author{Martin T. Luu \footnote{Martin T. Luu, Department of Mathematics, University of California, Davis, CA 95616, USA, E-mail address: mluu$\textrm{@}$math.ucdavis.edu}}

\date{}

\begin{document}

\maketitle

\begin{abstract}
We describe a simple algebraic approach to several spectral duality results for integrable systems and illustrate the method for two types of examples: The Bertola--Eynard--Harnad spectral duality of the two-matrix model as well as the various dual descriptions of minimal model conformal field theories coupled to gravity. 
\end{abstract}

\section{Introduction}

Adams, Harnad, and Hurtubise developed in \cite{AHH} a general framework to deal with integrable systems having Lax descriptions of different sizes. Originally this was done within the context of isospectral deformations, later on within the context of isomonodromic deformations. The latter circle of ideas is often called Harnad duality, see for example Boalch's work \cite{BOA} for a very general result in this direction. One aspect of these duality results for Lax operators is a duality of underlying spectral curves. Starting with \cite{BEH1}, and then continued in \cite{BEH2}, Bertola, Eynard, and Harnad have shown such spectral dualities for the two-matrix model. 

Fixing $N\ge 1$ and two polynomials $V_{1}(x)$ and $V_{2}(y)$, the corresponding two-matrix model of $N \times N$ matrices can be studied via biorthogonal polynomials $\pi_{n}(x)$ and $\sigma_{n}(x)$ which satisfy
$$\iint  \pi_{n}(x)\sigma_{m}(y)e^{-V_{1}(x)-V_{2}(y)+xy} \; \;  \textrm{d}x \;\textrm{d} y = \delta_{m,n}.$$

The spectral duality associated to the $\pi_{n}$'s (a similar result holds for the $\sigma_{n}$'s) of Bertola, Eynard, Harnad relates the $\partial_{x}$ action on the suitably normalized $\pi_{n}$'s to the analogous differentiation action on the Fourier transforms of the functions. When the size $N$ of the matrices involved is approaching $\infty$, the duality simplifies. In particular, it is shown in \cite{BEH2} that the duality can be viewed as two different ways of expressing a certain resultant of two Laurent polynomials $f,g$ in $\CC[\lambda,\lambda^{-1}]$, where $\lambda$ is some indeterminate. 

In the present work we take this algebraic viewpoint a bit further: One can view the ring of Laurent polynomials as a module over a polynomial ring $\CC[u]$ in two ways, by letting $u$ act via multiplication by $f$ or $g$. From this point of view, it is natural to generalize the duality of \cite{BEH2} to more general pairs of $\CC[u]$-module structures. We carry out this simple idea in Section \ref{spectral-duality-section} and essentially re-derive the Bertola--Eynard--Harnad duality in Section \ref{large-N-section}. 

The generality of the algebraic result of Section \ref{spectral-duality-section} suggests the possibility of developing a common framework for various spectral curve dualities. We illustrate this by showing that not only the above mentioned two-matrix duality can be deduced, but also the spectral duality aspect of the following duality of conformal field theories.

The minimal model conformal field theories coupled to gravity have two different mathematical descriptions via generalized KdV hierarchies. Write the central charge $c$ as
$$c=1-6 \cdot \frac{(p-q)^{2}}{pq}$$ 
for positive co-prime integers $p,q$. The theory is then known, see for example \cite{FKN}, to be describable by a $p$-reduced Lax operator of the KP integrable hierarchy as well as by a $q$-reduced Lax operator. There is a duality relating the two theories and as a classical limit there is a duality of spectral curves. We show in Section \ref{quantization-section}, see Theorem \ref{quantum-curve-duality} for a precise statement, that this spectral duality is just another example of the same algebraic spectral duality from which we deduce the two-matrix model duality.

\section{Spectral duality}
\label{spectral-duality-section}

We describe a very basic spectral duality for bundles on $\mathbb{A}^{1} = \textrm{Spec } \CC[u]$. The algebraic set-up for our approach is the following: 

Let $\mathcal V$ be a non-zero $\CC$-vector space with a $\CC$-linear endomorphism $A$. Let $u$ be an indeterminate and suppose $\mathcal V$ has the structure of a free finite rank $\CC[u]$-module by letting $u$ act via $A$. Let $a$ denote the rank and choose a $\CC[u]$-basis $\{v_{1},\cdots,v_{a}\}$. 

\begin{defi}
For a $\CC$-linear endomorphism $B$ of $\mathcal V$ let $M_{B,A}$ be the element of $\mathfrak g \mathfrak l_{a}[u]$ which describes the $B$ action with respect to $A$: For each $1\le i \le a$ let
$$B \cdot v_{i} = \sum_{j=1}^{a} (M_{B,A})_{j,i}(A) \;  v_{j}.$$
\end{defi}

\begin{defi}
Define the spectral curve
$$X_{B,A}= \left \{(x,y) \in \mathbb{A}^{2}(\CC) \; \Big |  \; \det(y\cdot \textrm{\textbf{1}}_{a} - M_{B,A}(x)) = 0 \right \}.$$
\end{defi}

Suppose now $P$ and $Q$ are two $\CC$-linear endomorphisms of $\mathcal V$ satisfying the conditions of $A$ above. This yields two spectral curves $X_{P,Q}$ and $X_{Q,P}$ and it turns out, see Lemma \ref{duality-lemma}, that these are related via a simple duality. A point we would like to emphasize is that this duality does not depend on any assumptions concerning the value of the commutator $[P,Q]$, in particular the duality is not restricted to the string equation context of \cite{BEH2} where $[P,Q]$ is a constant.
 
\begin{defi}
For a curve $X$ in $\mathbb{A}^{2}=\textrm{Spec }\CC[x,y]$ cut out by an equation $f(x,y)=0$ denote by $\check X$ the curve cut out by $f(y,x)=0$.
\end{defi}
The following simple result is central to our considerations:
\begin{lem}[Spectral Duality Lemma]
\label{duality-lemma}
There is an equality of spectral curves
$$X_{Q,P}= \check X_{P,Q}.$$
\end{lem}
\begin{proof}
Let $p$ and $q$ denote the rank of $\mathcal V$ when viewed as a $
\CC[u]$-module by letting $u$ act via $P$ and $Q$, respectively. Let 
$v_{1},\cdots,v_{p}$ and $w_{1},\cdots,w_{q}$ be corresponding $
\CC[u]$-bases and let $\mathcal V^{(1)}$ denote the $\CC$-span of the 
$v_{i}$'s and $\mathcal V^{(2)}$ the $\CC$-span of the $w_{i}$'s.

The curve $X_{Q,P}$ is the vanishing locus of $\det(y\cdot \textrm{\textbf{1}}_{p} - M_{Q,P}(x))$ and $\check X_{P,Q}$ is the vanishing locus of $\det(x\cdot \textrm{\textbf{1}}_{q} - M_{P,Q}(y))$. We prove their equality along similar lines as the arguments of \cite{BEH2} (Section 3). Define for $x,y$ in $\CC$ the $\CC$-linear map
$$\xi_{x,y} : \mathcal V^{(1)} \times \mathcal V^{(2)} \rightarrow \mathcal V$$
by
$$(v,w) \mapsto (Q- y \cdot \textrm{id}) \cdot v + (P-x \cdot \textrm{id}) \cdot w.$$
The map $\xi_{x,y}$ has a non-trivial kernel if and only if there is $v$ in $\mathcal V^{(1)}$ and $w$ in $\mathcal V^{(2)}$ such that
\begin{eqnarray}
\label{resultant-equation}
(Q- y \cdot \textrm{id}) \cdot v = (P-x \cdot \textrm{id}) \cdot w
\end{eqnarray}
and such that not both of $v$ and $w$ are $0$. Since both $\CC[u]$-module structures are free, it follows that neither $Q- y \cdot \textrm{id}$ nor $P-x \cdot \textrm{id}$ can have non-trivial kernel. Therefore, neither $v$ nor $w$ are $0$. Rewrite Equation (\ref{resultant-equation}) now in matrix form: Let $\overline{v} \in \CC^{p}$ denote the vector describing $v$ with respect to $v_{1},\cdots, v_{p}$ and let $\overline{w} \in \CC^{q}$ denote the vector describing $w$ with respect to $w_{1},\cdots,w_{q}$. Define a $\CC$-linear map $\phi$ from $\mathcal V$ to $\mathcal V^{(1)}$ by
$$\sum_{i=1}^{p} a_{i}(P) v_{i} \mapsto \sum_{i=1}^{p} a_{i}(x) v_{i}$$
where the $a_{i}$'s are polynomials. Applying $\phi$ to Equation (\ref{resultant-equation}) then gives
\begin{eqnarray}
\label{first-eigenvector-equation}
(M_{Q,P}(x)-y \cdot \textbf{1}_{p})  \; \overline{v}= \overline{0}.
\end{eqnarray}
An analogous argument yields
\begin{eqnarray}
\label{second-eigenvector-equation}
(M_{P,Q}(y)-x\cdot \textbf{1}_{q}) \; \overline{w} = \overline{0}.
\end{eqnarray}
The reverse direction of the argument also holds: Suppose Equation (\ref{first-eigenvector-equation}) has a non-zero solution $\overline{v}$. Then for the corresponding $v$ one knows that $(Q-y\cdot \textrm{id})\cdot v$ is an element $\sum_{i=1}^{p}a_{i}(P)v_{i}$ of $\mathcal V$ that is in the kernel of $\phi$. Hence $x$ is a root of each of the polynomials $a_{i}$ and hence $a_{i}(P)=(P-x)\cdot b_{i}(P)$ for some polynomial $b_{i}$ and hence 
$$(Q-y\cdot \textrm{id})\cdot v=(P-x\cdot \textrm{id}) \cdot \left( \sum_{i=1}^{p} b_{i}(P) v_{i} \right )$$
as desired. It follows that the kernel of $\xi_{x,y}$ is non-zero and that Equation (\ref{second-eigenvector-equation}) also has a non-zero solution. The same argument applies with the role of the two equations reversed. In conclusion, we have shown that the two characteristic polynomials defining the spectral curves $X_{Q,P}$ and $\check X_{P,Q}$ have the same vanishing locus in $\mathbb{A}^{2}(\CC)$.
\end{proof}

Note the following relation to resultants: Consider the special case of the set-up of Lemma \ref{duality-lemma} where $\mathcal V=\CC[\lambda]$ and $P,Q$ are the multiplication operators by elements in $\CC[\lambda]$ of degree $p$ and $q$ respectively. One can take $$\mathcal V^{(1)}= \textrm{$\CC$-span } \{1,\cdots, \lambda^{p-1}\} \;\; , \;\; \mathcal V^{(2)}=\textrm{$\CC$-span } \{1,\cdots,\lambda^{q-1}\}.$$ 
The determinant of $\xi_{x,y}$, viewed as a map from $\mathcal V^{(1)} \times \mathcal V^{(2)}$ to the span of $\{1,\cdots, \lambda^{p+q-1}\}$, is in fact the resultant of the polynomials $P-x$ and $Q-y$. The two characteristic polynomials defining $X_{Q,P}$ and $\check X_{P,Q}$ then correspond to two different ways to express this resultant. See for example \cite{GKZ} (Section 3.4) for how this fits into a very broad framework of calculation of generalized resultants. A similar remark applies when $\CC[\lambda]$ and $P,Q$ are changed to Laurent polynomials, as discussed by Bertola--Eynard--Harnad in \cite{BEH2}, and this remark is relevant for the calculations in Section \ref{large-N-section}.

\section{Applications}
We now give two applications of the duality described in the previous section. 

\subsection{Quantization of differential operators}
\label{quantization-section}

In this section we give an example of the set-up of the previous section
and relate the spectral duality to a (Fourier) duality of quantum curves. 

Recall that the operators $P$ and $Q$ of our constructions are related to $\CC[u]$-modules, where $u$ is our chosen coordinate on $\mathbb{A}^{1}$. Fix now two further indeterminates $z$ and $s$. Let $\mathcal V=\CC[z]$ and $v_{i}=z^{i-1}$ for $i=1,2,\cdots$. Consider the action of $\CC[\![s]\!][\partial_{s}]$ on $\mathcal V$ given via $\partial_{s} \mapsto z$ and $s \mapsto - \partial_{z}$. Let $P,Q$ be two elements of $\CC[\![s]\!][\partial_{s}]$. We assume $P$ and $Q$ are of positive
$\partial_{s}$-degree $p$ and $q$, respectively, with constant leading order coefficient. Hence for example
$$P= \textrm{const} \cdot \partial_{s}^{p} + \sum_{i=0}^{p-1} p_{i}(s) \partial_{s}^{i}$$ 
with $p_{i}(s) \in \CC[\![s]\!]$. We define two $\CC[u]$-module structures $\mathcal V_{1}$ and $\mathcal V_{2}$ on $\mathcal V$ by letting $u$ act via $P$ and $Q$, respectively. Then one sees that these are free modules with
$$p = \textrm{rank } \mathcal V_{1} \;\; , \; \; q = \textrm{rank } \mathcal V_{2}.$$  
For the $\CC$-spans $\mathcal V^{(1)}$ and $\mathcal V^{(2)}$ of $\CC[u]$-bases one can take 
\begin{eqnarray}
\label{quantization-choice-of-basis}
\mathcal V^{(1)} =  \textrm{$\CC$ - span} \; \{1,\cdots,z^{p-1} \} \;\; ,\;\;
\mathcal V^{(2)}  =  \textrm{$\CC$ - span} \; \{1, \cdots,z^{q-1} \}.
\end{eqnarray}

All conditions of the set-up of Section \ref{spectral-duality-section} are met and hence one obtains from Lemma \ref{duality-lemma} the following result:

\begin{thm}
\label{operator-duality}
For $P$ and $Q$ in $\CC[\![s]\!][\partial_{s}]$ as above, one has
$$X_{Q,P}= \check X_{P,Q}.$$
\end{thm}

We make some comments about special cases of this result:

\begin{example}
Fix positive co-prime integers $p$ and $q$. Write $q=a\cdot p +r$ for $0\le r <p$. Define $M_{q,p}$ in $\mathfrak g \mathfrak l_{p}[u]$ via
$$(M_{q,p})_{i,j}= \begin{cases} u^{a} \cdot \delta_{j+r,j} \;\; \;\;\;\;\;\;\;\;\;\;\; \textrm{ if } \;\; i \le p-r \\  \\ u^{a+1} \cdot \delta_{j+r-p+1,j} \;\; \textrm{ if } \;\;  p-r < i \le p \end{cases}$$

For the choice of $\mathcal V^{(1)}$ and $\mathcal V^{(2)}$ as described earlier one has $M_{\partial_{s}^{q},\partial_{s}^{p}}=M_{q,p}$ and Theorem \ref{operator-duality} implies that up to  the exchange of $x$ and $y$ the characteristic polynomials of $M_{p,q}$ and $M_{q,p}$ have the same vanishing locus. Of course this can be verified directly by noting that the eigenvalues of $M_{q,p}(x)$ are of the form $\zeta_{p}^{i} \cdot x^{q/p}$ with $1\le i \le p$, for a primitive $p$'th root of unity $\zeta_{p}$, and the eigenvalues of $M_{p,q}(y)$ are of the form $\zeta_{q}^{i} \cdot y^{p/q}$ with $1\le i \le q$, for a primitive $q$'th root of unity $\zeta_{q}$. Hence, up to the exchange of $p$ and $q$ the spectral curves are the vanishing locus of $y^{q}-x^{p}$.
\end{example}

\begin{example}
Consider an example of $P$ and $Q$ whose commutator is not a constant. For example, let $P=\partial_{s}^{2}+s+1$, $Q=\partial_{s}^{3}-2$. Then
$$M_{Q,P}= \begin{bmatrix}
-1 & (u-1)^{2}\\
u-1 & 0
\end{bmatrix} \; \; , \; \; M_{P,Q} = \begin{bmatrix}
1 & 1+u & 0 \\
0 & 1 & u \\
1&0&1
\end{bmatrix}$$
and the spectral duality can be verified directly:
$$\det(y \cdot \textbf{1}_{2}-M_{Q,P}(x))=y^{2}  + y -(x-1)^{3}$$
$$\det(x \cdot \textbf{1}_{3}-M_{P,Q}(y))=(x-1)^{3} - y^{2}-y.$$
\end{example}

After giving these two examples of Theorem \ref{operator-duality}, we now use the spectral duality to describe a duality in Schwarz's approach, see \cite{SCH}, to quantum curves. 

The classical data is an ordered pair $(P_{0},Q_{0})$  of commuting elements of $\CC[\![s]\!][\partial_{s}]$. Following Schwarz \cite{SCH} we define:
\begin{defi}
\label{quantization-definition}
A pair $(P_{1},Q_{1})$ of elements in $\CC[\![s]\!][\partial_{s}]$ is a quantization of $(P_{0},Q_{0})$ if 
\begin{enumerate}[(i)]
\item
$\textrm{deg}_{\partial_{s}} P_{1} = \textrm{deg}_{\partial_{s}} P_{0} \; \; , \;\;  \textrm{deg}_{\partial_{s}} Q_{1} = \textrm{deg}_{\partial_{s}} Q_{0}$
\item
$[P_{1},Q_{1}]=1$
\item
$
M_{P_{1},Q_{1}} = M_{P_{0},Q_{0}}$ (with respect to the choice of bases in Equation (\ref{quantization-choice-of-basis})).
\end{enumerate}
\end{defi}
The characteristic polynomial of $
M_{P_{1},Q_{1}} = M_{P_{0},Q_{0}}$ gives rise to a curve in $\mathbb{A}^{2}$. We call this the spectral curve and denote it by $X_{(P_{1},Q_{1})}$.
\begin{defi}
If the equality in condition (iii) is replaced by mere spectral equivalence (meaning same vanishing locus of characteristic polynomial), we call $(P_{1},Q_{1})$ a spectral quantization of $(P_{0},Q_{0})$.
\end{defi}

\begin{remark}
While this might not be immediate from the definition, the notion of quantization of differential operators is useful in quantum field theory: 

The quantization of $(\partial_{s}^{p},\partial_{s}^{q})$ leads to a pair $(P_{1},Q_{1})$ such that if $\gamma$ is a monic degree $0$ pseudo-differential operator with 
$$\gamma P_{1} \gamma^{-1} = \partial_{s}^{p}$$ 
then $\gamma \cdot \CC[z]$ is a point of the big cell of the Sato Grassmannian whose associated KP tau function is (essentially) the partition function of the $(p,q)$ minimal model coupled to gravity, generalizing the Witten--Kontsevich tau function which corresponds to $(p,q)=(2,1)$. We refer to \cite{SCH} for details.
\end{remark}

The commutation relation among $P_{1}$ and $Q_{1}$ in Definition \ref{quantization-definition} is sometimes called the string equation. Note that these come in pairs:
$$[P_{1},Q_{1}]=1 \;\; ,\;\; [-Q_{1},P_{1}]=1.$$

These can be viewed as Fourier dual string equations, coming from the Fourier transform of the one variable Weyl algebra $\CC[u,\partial_{u}]$ given by $u \mapsto \partial_{u} $ and $\partial_{u} \mapsto -u$. We denote this Fourier dual pair of operators by
$$\mathcal F(P_{1},Q_{1}) = (-Q_{1},P_{1}).$$
The following construction can be thought of as an abelian analogue of this Fourier transform.
\begin{defi}
For a curve $X$ in $\mathbb{A}^{2}=\textrm{Spec }\CC[x,y]$ cut out by an equation $f(x,y)=0$ denote by $\mathcal F X$ the curve cut out by $f(-y,x)=0$.
\end{defi}
We now show that the spectral duality of Theorem \ref{operator-duality} also implies a duality for the quantization of differential operators: 
\begin{thm}
\label{quantum-curve-duality}
Suppose $(P_{1},Q_{1})$ is a quantization of $(P_{0},Q_{0})$. Then 
$$X_{\mathcal F(P_{1},Q_{1})}=\mathcal F \left ( X_{(P_{1},Q_{1})}\right )$$
and $\mathcal F(P_{1},Q_{1})$ is always a spectral quantization of $\mathcal F(P_{0},Q_{0})$ but in general it is not a quantization.
\end{thm}
\begin{proof}
We write $\sim$ if two elements in $\CC[x,y]$ have the same vanishing locus. It follows from the Spectral Duality Lemma \ref{duality-lemma} that 
\begin{eqnarray*}
\det(y \cdot \textbf{1}_{p} - M_{-Q_{1},P_{1}}(x)) & =& (-1)^{p} \cdot  \det((-y) \cdot \textbf{1}_{p} - M_{Q_{1},P_{1}}(x) )\\[5pt]
&\sim &\det(x \cdot \textbf{1}_{q} - M_{P_{1},Q_{1}}(-y) )\\[5pt]
&= & \det(x \cdot \textbf{1}_{q} - M_{P_{0},Q_{0}}(-y) )\\[5pt]
&\sim & \det((-y) \cdot \textbf{1}_{p} - M_{Q_{0},P_{0}}(x) )\\[5pt]
&=&  (-1)^{p} \det(y \cdot \textbf{1}_{p} - M_{-Q_{0},P_{0}}(x) ). 
\end{eqnarray*}
It follows from the second line that $X_{\mathcal F(P_{1},Q_{1})}=\mathcal F \left ( X_{(P_{1},Q_{1})} \right )$ and it follows from the fifth line that $\mathcal F(P_{1},Q_{1})$ is a spectral quantization of $\mathcal F(P_{0},Q_{0})$. 

We now prove the second part of the theorem that in general $\mathcal F(P_{1},Q_{1})$ is only a spectral quantization of $\mathcal F(P_{0},Q_{0})$: Consider for example
$$P_{0}=\partial_{s}^{2}-2\partial_{s}+1 \;\; , \;\; Q_{0}=\partial_{s}.$$
One sees that $(P_{1},Q_{1})$ is a quantization, where
$$P_{1}=\partial_{s}^{2}-s \;\; , \;\; Q_{1}=\partial_{s}+1.$$
However,
$$M_{Q_{0},P_{0}}= \begin{bmatrix}
0 & u-1 \\
1& 2
\end{bmatrix} \ne \begin{bmatrix}
1 & u \\ 1 & 1
\end{bmatrix}=M_{Q_{1},P_{1}}$$
even though they have the same spectrum.
\end{proof}
\begin{rem}
The above Fourier duality in the quantization of differential operators is in fact related to the $x-y$ symmetry in topological recursion \cite{EO1}, \cite{EO2}. The latter is modeled on two-matrix model calculations. We discuss a spectral duality for the two-matrix model in the next section. 
\end{rem}

Theorem \ref{operator-duality} relates the spectra of $M_{Q,P}$ and $M_{P,Q}$. Sometimes more can be said and the two matrices completely determine each other. In the remainder of this section we describe this for the example of the matrices $M_{p,q}$ (as remarked before, their quantization is related to the $(p,q)$ minimal model conformal field theory coupled to gravity). Assume now there are $\CC$-linear endomorphisms $P$ and $Q$ of a $\CC$-vector space $\mathcal V$ satisfying the conditions in Section \ref{spectral-duality-section} (note that we do not restrict here to the case $\mathcal V = \CC[z]$) and let $p$ and $q$ denote the corresponding $\CC[u]$-module ranks. Assume that there are elements $v_{1},\cdots $ in $\mathcal V$ such that $\mathcal V^{(1)}$ is the span of $v_{1},\cdots,v_{p}$ and $\mathcal V^{(2)}$ is the span of $v_{1},\cdots,v_{q}$. 

Assume now that $M_{Q,P}=M_{q,p}$ . We claim that if $p>q$ then this automatically implies that $M_{P,Q}=M_{p,q}$, while the reverse implication does not hold in general. We give the details in the illustrative example $\{p,q\}=\{2,3\}$. One has
$$M_{2,3}= \begin{bmatrix}
0 & u & 0 \\
0&0&u\\
1&0&0
\end{bmatrix} \; , \; M_{3,2} = \begin{bmatrix}
0 & u^{2}\\
u&0
\end{bmatrix}.$$
If the operators $P$ and $Q$ satisfy $M_{Q,P}=M_{2,3}$ then in particular
$$Q\cdot v_{1}=v_{3} \;\; , \;\; Q\cdot v_{2}=P\cdot v_{1} \;\; , \;\;  Q \cdot v_{3} = P \cdot v_{2}$$ 
Therefore, since $P\cdot v_{1} = Q\cdot v_{2}$ and $P\cdot v_{2} = Q \cdot v_{3} = Q\cdot (Q\cdot v_{1})$ it follows that $M_{P,Q}=M_{3,2}$, as desired. 

Suppose now that there are operators $P$ and $Q$ with $M_{P,Q}=M_{3,2}$. Hence
$$P\cdot v_{1} =Q\cdot v_{2} \;\; , \;\; P \cdot v_{2}= Q^{2} \cdot v_{1}.$$
Then for any choice of $\xi \in \CC$ the following are valid choices for $M_{Q,P}$:
$$\begin{bmatrix}
0 & u & - \xi u\\
\xi & 0 & u \\
1&0 &0
\end{bmatrix}.$$
In particular, if $\xi \ne 0$ then $M_{Q,P} \ne M_{2,3}$. Note however that as implied by the Spectral Duality Lemma \ref{duality-lemma}, the spectrum is independent of $\xi$: The characteristic polynomial is $\lambda^{3}-u^{2}$ and is obtained (up to a sign) from the spectrum of $M_{3,2}$ with characteristic polynomial $\lambda^{2}-u^{3}$ via the exchange of the variables $\lambda$ and $u$.

\subsection{Spectral duality for two-matrix models}
\label{large-N-section}

Bertola, Eynard, and Harnad prove, among other things, in \cite{BEH1}, \cite{BEH2} a spectral duality for the two-matrix model. In the current section we revisit their result from the point of view described in Section \ref{spectral-duality-section}.

Let $x$ and $y$ be two indeterminates and fix two polynomials $V_{1}(x),V_{2}(y)$ of positive degree $d_{1}+1$ and $d_{2}+1$, respectively. For any fixed integer $N \ge 1$ one can associate to these polynomials a two-matrix model (involving pairs of $N\times N$ Hermitian matrices) via a probability measure depending on $V_{1}$ and $V_{2}$. It is known that the study of the two-matrix model (for whatever choice of $N$) can be reduced to a study of a certain collection of polynomials, namely the sequences $\{\pi_{n}(x)\}$ and $\{\sigma_{n}(y)\}$ of degree $n$ biorthogonal polynomials, defined via 
$$\iint \pi_{n}(x)\sigma_{m}(y)e^{-V_{1}(x)-V_{2}(y)+xy} \;\; \textrm{d}x \; \textrm{d} y = \delta_{m,n}.$$
We refer to \cite{BEH1} for more details, in particular for a discussion and references concerning choices of integration contour.

The spectral duality that we would like to discuss has a version for the $\pi_{n}$'s as well as the $\sigma_{n}$'s. We focus on the first case here for brevity, the other case can be treated in a completely analogous manner. It is useful to normalize the $\pi_{n}$'s in the following manner: For each $n \ge 0$ let
$$
\psi_{n}(x)  =   \pi_{n}(x)e^{-V_{1}(x)} 
$$
and consider the Fourier-Laplace transforms
$$\widehat{\psi}_{n}(y)  :=  \int \textrm{d}x  \; e^{xy}\psi_{n}(x).$$
We refer again to \cite{BEH1} for a discussion of integration contours.

The spectral duality concerns a relation between the $\partial_{x}$ and $\partial_{y}$ action on these Fourier dual sequences of functions. To state it precisely, fix now a positive integer $N\ge d_{2}$ and let 
$$\Psi := [\psi_{N-d_{2}}, \cdots, \psi_{N}]^{\textrm{T}}$$ 
and 
$$\widehat{\Psi} := [\widehat{\psi}_{N-1},\cdots,\widehat{\psi}_{N-1+d_{1}}].$$ 
As discussed in \cite{BEH2}, there is $M_{2}(x)$ in $\mathfrak g\mathfrak l_{d_{2}+1}[x]$ with  
$$\partial_{x} \; \Psi = - M_{2}(x) \cdot  \Psi$$
and there is $M_{1}(y)$ in $\mathfrak g \mathfrak l_{d_{1}+1}[y]$ with
$$\partial_{y} \; \widehat{\Psi}= \widehat{\Psi} \cdot M_{1}(y).$$
Bertola, Eynard, and Harnad prove in \cite{BEH2} (Proposition 4.1) the following spectral duality:
$$\det (x\cdot \textbf{1}_{d_{1}+1} - M_{1}(y)) = * \cdot \det (y \cdot \textbf{1}_{d_{2}+1}-M_{2}(x))$$
for some (explicitly calculable) non-zero scalar $*$. 

Recall that the integer $N$ was fixed in the above constructions. As Bertola--Harnad--Eynard have shown in \cite{BEH2} (Section 3), the spectral duality has a large $N$ limit analogue. It is this latter version of the spectral duality (see Theorem \ref{BHN-theorem} for a precise statement) that we now wish to relate to the spectral duality described in Lemma \ref{duality-lemma}.

Let $\lambda$ be an indeterminate and consider two elements of $\CC[\lambda,\lambda^{-1}]$ given by 
\begin{eqnarray}
\label{P-equation}
P = \gamma \cdot \lambda^{-1}  + \sum_{i=0}^{d_{1}} b_{i} \lambda^{i}
\end{eqnarray}
and 
\begin{eqnarray}
\label{Q-equation}
Q= \gamma \cdot \lambda + \sum_{i=0}^{d_{2}} a_{i}\lambda^{-i}
\end{eqnarray}
with $\gamma \cdot a_{d_{2}} \cdot b_{d_{1}} \ne 0$ and $d_{1}>0$ and $d_{2}>0$. Bertola--Harnad--Eynard define in \cite{BEH2}
\begin{eqnarray}
\label{A-equation}
A&=&\begin{bmatrix}
0 & 1 & 0 & \hdots & 0\\
0&0&1 & \hdots & 0\\
\vdots&&&\ddots& \vdots\\
0&& & \hdots & 1\\
-a_{d_{2}}/\gamma &  \hdots& &-a_{1}/\gamma & (x - a_{0})/\gamma
\end{bmatrix} \\
B&=&\begin{bmatrix}
0 & 1 & 0 & \hdots & 0\\
0&0&1 & \hdots & 0\\
\vdots&&&\ddots& \vdots\\
0&& & \hdots & 1\\
-b_{d_{1}}/\gamma &  \hdots& &-b_{1}/\gamma & (y - b_{0})/\gamma
\end{bmatrix}
\end{eqnarray}
as well as
$$D_{1}(x)=\gamma A^{-1} + \sum_{i=0}^{d_{1}} b_{i} A^{i} \;\;\; , \;\;\; D_{2}(y)=\gamma B+ \sum_{i=0}^{d_{2}} a_{i} B^{-i}.$$
The following is shown in \cite{BEH2} (Section 3): 
\begin{thm}[Bertola--Eynard--Harnad]
\label{BHN-theorem}
There is a non-zero constant $*$ such that
$$\det(y \cdot \emph{\textbf{1}}_{d_{2}+1}-D_{1}(x)) =* \cdot  \det (x \cdot \emph{\textbf{1}}_{d_{1}+1}  - D_{2}(y)).$$
\end{thm}

We give a quick proof of this using our spectral duality result from Section \ref{spectral-duality-section}.

Consider the space $\mathcal V= \CC[\lambda,\lambda^{-1}]$ of Laurent polynomials and consider a $\CC[u]$-module structure on $\mathcal V$ where $u$ acts via multiplication by $R = \sum_{i=e}^{f} c_{i}\lambda^{i}$ with $e<0$ and $f>0$ and $c_{e}\cdot c_{f} \ne 0$. By a degree consideration this is a torsion-free module. It is also finitely generated and in fact we claim that $\mathcal B=\{\lambda^{e}, \cdots, \lambda^{f-1}\}$ is a basis. Note that $\mathcal B$ is linearly independent over $\CC[u]$: Otherwise, there is a non-zero $\CC$-linear combination of the elements of $\mathcal B$ that equals $u \cdot \xi $ for some $\xi$ in $\mathcal V$ and a degree consideration gives a contradiction. Furthermore, $\mathcal B$ is a generating set: By using a change of coordinates $\lambda \mapsto \lambda^{-1}$ argument one sees that it is sufficient to show that each $\lambda^{i}$ for $i \ge 0$ can be obtained. By induction, suppose $\lambda^{j}$ for $0\le j <i$ is already known to be obtainable. Write $i = q f + r$ with $0 \le r < f$. Then
$R^{q}\lambda^{r} =C_{1}\lambda^{i}+ \cdots +C_{2}\lambda^{qe+r}$ for some non-zero constants $C_{1},C_{2}$. Considering $R^{a}\lambda^{b}$ with $0\le a \le q-1$ and $r+e \le b \le r$ (note $\lambda^{b} \in \mathcal B$) one obtains expressions with lowest terms ranging from $qe+r$ till $r+e$ and with highest term strictly less than $i$. Hence, some linear combinations of these terms with $R^{q}\lambda^{r}$ has no terms below $\lambda^{e}$. It follows that $\mathcal B$ is a generating set. In conclusion, one obtains that $\mathcal V$ is a free $\CC[u]$-module of rank $|e|+f$. In fact, any $|e|+f$ consecutive powers of $\lambda$ are linearly independent and form a basis. 

We now apply the above argument to the two $\CC[u]$-module structures $\mathcal V_{1}$ and $\mathcal V_{2}$ on $\mathcal V=\CC[\lambda,\lambda^{-1}]$ where the action of $u$ is via
$P$ and $Q$, respectively, as defined in Equation (\ref{P-equation}) and Equation (\ref{Q-equation}). It follows that
\begin{eqnarray*}
p &:= &\textrm{rank } \mathcal V_{1}=d_{1}+1 \\
 q &:= & \textrm{rank } \mathcal V_{2}= d_{2}+1
\end{eqnarray*}
and one can make the following choices
\begin{eqnarray*}
\mathcal V^{(1)} & =&  \textrm{$\CC$ - span} \; \{\lambda^{N-1}, \cdots, \lambda^{N-1+d_{1}} \} \\[5pt]
\mathcal V^{(2)} & =& \textrm{$\CC$ - span} \; \{\lambda^{N-d_{2}}, \cdots, \lambda^{N} \}.
\end{eqnarray*}
The Spectral Duality Lemma \ref{duality-lemma} can be applied and it follows that the vanishing locus in $\mathbb{A}^{2}(\CC)$ of $\det(y \cdot \textbf{1}_{d_{2}+1}-M_{P,Q}(x))$ and of $\det (x \cdot\textbf{1}_{d_{1}+1}  - M_{Q,P}(y))$ agree. To relate this to Theorem \ref{BHN-theorem} observe the following: The matrix $A$ of Equation (\ref{A-equation})
is the description of multiplication by $\lambda$ on $\lambda^{N-d_{2}},\cdots, \lambda^{N}$ if the $x$ action is given via $Q$. Hence, $D_{1}(x)$ is the matrix description of the $P$-action with respect to $Q$ in this same basis:
$$D_{1}(x) = M_{P,Q}(x).$$ 
Similarly one obtains
$$D_{2}(y)=M_{Q,P}(y).$$
Therefore, as claimed, the duality result of Bertola--Eynard--Harnad follows (up to questions of multiplicity of vanishing) from the spectral curve duality described in Section \ref{spectral-duality-section}.

\subsubsection*{Acknowledgments:}
It is a great pleasure to thank John Harnad for very informative exchanges.

\end{document}